\documentclass[10pt]{article}
\usepackage[utf8]{inputenc}

\usepackage{adjustbox}

\usepackage{amssymb,amsmath,amsthm}
\usepackage{authblk}
\usepackage{bm}
\usepackage{breqn}
\usepackage[justification=centering]{caption}
\usepackage{color}
\usepackage{comment}
\usepackage{dsfont}
\usepackage{enumitem}
\usepackage{geometry}
\usepackage{graphicx}
\usepackage{xurl} \usepackage{hyperref}
\usepackage{mathtools}
\usepackage[round]{natbib}
\usepackage{tikz}
\usetikzlibrary{positioning,fit,calc,backgrounds,arrows.meta}
\usepackage{subfig}

\hypersetup{
    colorlinks=true,
    pdftitle={A note on Bayes' rule, conditional densities, and recent criticisms of Bayesian inference}
    }

\newcommand{\I}{\mathds{1}}
\DeclareMathOperator{\E}{\mathbb{E}}
\newcommand{\der}{\operatorname{d\!}{}}

\let\P\relax
\DeclareMathOperator{\P}{\mathbb{P}}
\newcommand{\N}{\mathbb{N}}
\DeclareMathOperator{\Normal}{\mathcal{N}}

\renewcommand{\Re}{\mathbb{R}}

\DeclareMathOperator{\Exp}{\text{Exp}}
\DeclareMathOperator{\GammaDist}{\text{Gamma}}
\DeclareMathOperator{\Unif}{Unif}

\theoremstyle{plain}
\newtheorem{proposition}{Proposition}

\newcommand{\mybox}[1]{%
  \begin{center}
  \fbox{%
    \begin{minipage}{0.9\linewidth}{
    \footnotesize
    #1
    }
    \end{minipage}%
  }%
  \end{center}
}

\title{A note on conditional densities, Bayes' rule, and recent criticisms of Bayesian inference}
\author[1,2,*]{Alex Yan}
\author[1,2]{Cathal Mills}
\author[3]{Augustin Marignier}
\author[1,2]{Younjung Kim}
\author[1,2,*]{Ben Lambert}
\affil[1]{Department of Statistics, University of Oxford, Oxford, United Kingdom}
\affil[2]{Pandemic Sciences Institute, University of Oxford, Oxford, United Kingdom}
\affil[3]{Department of Earth Sciences, University of Oxford, Oxford, United Kingdom}
\affil[*]{alexander.yan@stats.ox.ac.uk, ben.lambert@stats.ox.ac.uk}
\date{\vspace{-5ex}}

\begin{document}

\maketitle


\begin{abstract}
When performing Bayesian inference, we frequently need to work with conditional probability densities. For example, the posterior function is the conditional density of the parameters given the data. Some might worry that conditional densities are ill-defined, considering that for a continuous random variable $Y$, the event $\{Y=y\}$ has probability zero, meaning the formula $\mathbb{P}(A|B)=\mathbb{P}(A\cap B)/\mathbb{P}(B)$ is inapplicable. 
In reality, when we work with conditional densities, we never condition directly on the zero-probability event $\{Y=y\}$; rather, we first condition on the random variable $Y$, and then we may plug in an observed value $y$. The first purpose of our article is to provide an exposition on conditional densities that elaborates on this point. While we have aimed to make this explanation accessible, we follow it with a roadmap of the measure theory needed to make it rigorous.
A recent preprint \citep[arXiv:2411.13570]{Mosegaard2024} has expressed the concern that probability densities are ill-defined and that as a result Bayes' theorem cannot be used, and they provide examples that allegedly demonstrate inconsistencies in the Bayesian framework. The second purpose of our article is to investigate their claims. We contend that the examples given in their work do not demonstrate any inconsistencies; we find that there are mathematical errors and that they deviate significantly from the Bayesian framework.
\end{abstract}
\tableofcontents

\section{Introduction}

Bayesian inference is a statistical approach to learning from data, in which we start with prior beliefs about a system and use observed data to update these beliefs in a principled way according to Bayes' rule. This allows us to fit our models to data and subsequently perform inference tasks like estimation and prediction, while ensuring that we get uncertainty quantification. Bayesian inference is used widely across scientific and data-driven disciplines, including to solve inverse problems in the physical sciences.

In a recent preprint, \cite{Mosegaard2024} claims to ``demonstrate mathematical inconsistency and logical acausality of commonly used Bayesian methods".  At first, their work (hereinafter ``MC") seems to have found fundamental flaws in the Bayesian paradigm, despite the fact that the Bayesian paradigm is a well-established approach for statistical inference. Our article examines the claims of MC and finds that Bayesian inference is free of the inconsistency and so-called ``logical acausality" that MC appears to have uncovered. 

A central claim of MC's manuscript is that conditional densities are ill-defined and that Bayes' theorem does not work in non-discrete spaces.\footnote{In their abstract, they assert that ``the notion of conditional densities is inadmissible" and that ``Bayes theorem itself is inadmissible in non-discrete spaces". In MC, the term “inadmissible” is not used in its standard technical sense from statistical decision theory. Likewise, “inconsistency” and “logical acausality” are not defined as formal concepts.} This claim is incorrect. Notions of conditioning are made rigorous with measure theory. Measure theory was developed in the first half of the 20th century; it also formalises the notions of probability and random variables. To support their claim, MC cites the Borel--Kolmogorov paradox \citep[also known as Borel's paradox;][]{borel1909,Kolmogorov1933}. The Borel--Kolmogorov paradox is of little relevance once we are equipped with a proper theory of conditional distributions. The Borel--Kolmogorov paradox concerns conditioning on an event of zero probability, whereas the theory of conditional distributions concerns conditioning on a random variable.

In Section \ref{section:2}, we shall explain that Bayes' theorem does work in non-discrete spaces. We shall provide an exposition explaining what conditional densities are, the conditions under which a conditional density exists (namely, that the joint density must exist) and to what extent it is uniquely defined. As part of this, in Section \ref{subsection:BK} we shall elaborate on why the Borel--Kolmogorov paradox is not a problem here.
We believe that Section \ref{section:2} may be of interest to a wider readership, so we have aimed for this section to make sense independently of the rest of this article. An understanding of measure theory is not necessary in order for us to work with conditional densities when performing Bayesian inference in practice, and correspondingly we have avoided involving measure theory until Section \ref{subsection:measuretheory}, where we sketch a roadmap of the relevant measure theory and point to textbook references, except for some technical clarifications that the reader may wish to ignore along the way.

In addition to the theoretical concern about whether conditional densities are well-defined, MC also provides a number of examples, detailed in their appendices, each of which appears to demonstrate inconsistencies that arise from Bayesian inference. In Sections \ref{section:3} and \ref{section:4}, we shall address these examples. Our findings are that: the example in their Appendix A deviates from the Bayesian framework in a fundamental way, making an invalid heuristic argument to find a conditional distribution; the examples in Appendices C and D each make the same critical mistake when attempting to reparametrise the data space; the examples in Appendix F do not perform Bayesian inference correctly; the example in Appendix G does not demonstrate any genuine inconsistency. Section \ref{section:3} has a loosely unifying theme that it is the modeller's job to choose appropriate models. Section \ref{section:4} conveys the message that Bayesian inference remains invariant to reparametrisations of the data space.

Finally, in Section \ref{section:5}, we shall mention some genuine advantages and disadvantages of using Bayesian inference. 
Like Section \ref{section:2}, this section may be read independently from the rest of the article.

\section{Conditional densities and Bayes' rule} \label{section:2}
\subsection{What is Bayes' rule?}
In the Bayesian framework, both parameters $\theta$ and data $Y$ are modelled as random. Their joint distribution is specified by a probability distribution $\pi(\theta)$ over the parameters, known as the prior distribution, and a conditional distribution for data given parameters, $f(y|\theta)$, known as the likelihood. Then, the conditional distribution for the parameters given the data can be obtained using the formula:\footnote{Here, the proportionality symbol denotes that we have dropped a normalising constant that depends only on $y$ and not on $\theta$.}
\begin{equation}
    \pi(\theta|y)\propto \pi(\theta)f(y|\theta),
    \label{eqn:bayesrule}
\end{equation}
and this is known as the posterior distribution.

The term ``Bayes' rule" refers to two concepts:\footnote{In addition to these, ``Bayes rule" (written without the apostrophe) may also refer to a Bayes decision rule, such as a Bayes classifier.}
\begin{enumerate}
    \item The methodological principle that the correct way to update our prior beliefs in light of new observations is to condition on the observations;
    \item A mathematical formula, such as (\ref{eqn:bayesrule}), that describes how conditioning works.
\end{enumerate}
We shall use the term ``Bayes' theorem" for the latter, and this will be the topic of our discussion.

In some texts, the single letter $p$ is used to denote all probability distributions, and Bayes' theorem may be written
$$p_{\theta|Y}(\theta|y)\propto p_{\theta}(\theta)p_{Y|\theta}(y|\theta).$$ We shall also use this practice. (Some texts omit the subscripts entirely and write $p(\theta|y)\propto p(\theta)p(y|\theta)$ for notational convenience. With this overloaded notation, which probability distribution $p$ is referring to is implied by the argument inside the brackets.)

An important point is that each of the functions involved in Equation (\ref{eqn:bayesrule}) could denote either a probability mass or a probability density, depending on whether their (first) argument is discrete or continuous.
\begin{itemize}
    \item If $\theta$ and $Y$ are discrete random variables, then $\pi(\theta)$ denotes the probability mass of $\theta$, $\pi(\theta|y)$ denotes the conditional probability mass of $\theta$ given $Y$, and $f(y|\theta)$ denotes the conditional probability mass of $Y$ given $\Theta$. We come to this in Section \ref{subsection:discrete}.
    \item If $\theta$ and $Y$ are continuous random variables, then $\pi(\theta)$ denotes the probability density of $\Theta$, $\pi(\theta|y)$ denotes the conditional probability density of $\Theta$ given $Y$ (when it exists), and $f(y|\theta)$ denotes the conditional probability density of $Y$ given $\Theta$ (when it exists). We come to this in Section \ref{subsection:continuous}.
    \item It may also be that $\theta$ or $Y$ are random vectors. The random vectors may have a combination of discrete and continuous elements. The dimension of $\theta$ may not be fixed (for instance, in a mixture model with a variable number of components that may contribute to the mixture), and the dimension of $Y$ may not be fixed (for instance, if certain event times are observed but the number of events is itself random). Nonetheless, conditional densities exist whenever corresponding joint densities exist. Then, Bayes' theorem still holds. We come to this in Section \ref{subsection:towards}.
\end{itemize}

Our exposition in this section justifies the way conditional densities and Bayes' theorem are used in practice. We hope this will reassure the reader that they may continue using Bayes' theorem without worrying about issues raised by MC.

\subsection{Conditioning on an event}
Let us begin with the concept of conditioning on an event. Conditioning on an event is well-defined as long as that event has positive probability.

Let $B$ be an event with positive probability, $\P(B)>0$. If $A$ is another event, then the conditional probability of $A$ given $B$ is the probability that both events occur divided by the probability of $B$:
\begin{equation}
    \P(A|B)=\frac{\P(A\cap B)}{\P(B)}.
    \label{eqn:condprobevents}
\end{equation}

A direct consequence of this definition is that if $A$ and $B$ both have positive probability, then
\begin{equation}
    \P(A|B)=\frac{\P(B|A)\P(A)}{\P(B)}.
    \label{eqn:bayesthmevents}
\end{equation}
This identity is a version of Bayes' theorem for events; it is distinct from versions of Bayes' theorem for random variables.

It should be stressed that $\P(A|B)$ is only defined when $B$ has positive probability, otherwise we end up with zero divided by zero on the right-hand side of (\ref{eqn:condprobevents}), which is nonsensical. The Borel--Kolmogorov paradox arises when one tries to condition on a zero-probability event. In essence, the paradox is that if $\P(B)=0$, one may try to use heuristic arguments to compute $\P(A|B)$, but different heuristic arguments will lead to different resulting values. One attempt at a definition might be to choose a sequence of nested positive-probability events $B_1\supset B_2\supset B_3\supset\dots$ that serve as better and better approximations to $B$, and define $\P(A|B)$ as the limit of conditional probabilities $\lim_{n\to\infty} \P(A|B_n)$. Unfortunately, in general this does not work, because different choices of $B_1,B_2,B_3,\dots$ may give different limits.

\subsection{Conditioning with discrete random variables} \label{subsection:discrete}
If $X$ and $Y$ are both discrete random variables, then the conditional distribution of $X$ given $Y$ is given by
\begin{equation}
    p_{X|Y}(x|y)=\frac{p_{X,Y}(x,y)}{p_Y(y)}.
    \label{eqn:conddistdiscrete}
\end{equation}
Here $p_{X,Y}$ and $p_Y$ denote probability mass functions, and this formula holds for any $y$ such that $p_Y(y)>0$. This is nothing new; if we like, we may think of (\ref{eqn:conddistdiscrete}) as notational shorthand for
$$
    \P(X=x|Y=y) = \frac{\P(X=x,Y=y)}{\P(Y=y)},
$$
which comes from (\ref{eqn:condprobevents}), bearing in mind that $\{X=x\}$ and $\{Y=y\}$ are both events.

It follows that for $y$ with $0<p_Y(y)<\infty,$
\begin{equation}
    p_{X|Y}(x|y)=
    \begin{cases}
        \frac{p_{Y|X}(y|x)p_X(x)}{p_Y(y)}, &\text{when } p_X(x)>0,\\
        0, &\text{otherwise}.
    \end{cases}
    \label{eqn:bayesthmdiscrete}
\end{equation}
This can be regarded as Bayes' theorem for discrete random variables. This is a simple restatement of (\ref{eqn:bayesthmevents}), but now we have the additional interpretation that for each $y$, we have specified a (conditional) distribution over $x$. The special case $p_X(x)=0$ may be understood implicitly\footnote{The numerator $p_{Y|X}(y|x)p_X(x)$ would be zero regardless of $p_{Y|X}(y|x)$.}; then noting that $p_Y(y)$ is just a normalising constant independent of $x$, we can also write (\ref{eqn:bayesthmdiscrete}) in the more compact form
\begin{equation}
    p_{X|Y}(x|y) \propto p_{Y|X}(y|x)p_X(x),
    \label{eqn:bayesrulediscrete}
\end{equation}
as we saw in Equation (\ref{eqn:bayesrule}).

\subsection{Conditioning with continuous random variables} \label{subsection:continuous}
Suppose now that $X$ and $Y$ are both continuous random variables on $\Re$. In this scenario, for any possible value $y$, the event $\{Y=y\}$ has probability zero. We cannot directly condition on zero-probability events like $\{Y=y\}$; nevertheless, we can condition on the random variable $Y$ to obtain the conditional density of $X$ given $Y$.

For this to work, we need that $X$ and $Y$ have a joint density $p_{X,Y}(x,y)$ on $\Re\times\Re$. Then, the marginal density of $Y$ can be obtained as $p_Y(y)=\int_\Re p_{X,Y}(x,y) \der x$, and the conditional density of $X$ given $Y$, denoted $p_{X|Y}$, can be defined by (\ref{eqn:conddistdiscrete})
whenever $0<p_Y(y)<\infty$.
This is to say that the formula for the continuous case is symbolically identical to the formula for the discrete case, but its interpretation is different, as now $p_Y$ and $p_{X,Y}$ are probability density functions, not probability mass functions.

It follows that we have a version of Bayes' theorem for continuous random variables, sometimes called the generalised Bayes' theorem, and the formula is exactly the same as before (\ref{eqn:bayesthmdiscrete}, \ref{eqn:bayesrulediscrete}), but its interpretation is different.

\paragraph{Example: exponential waiting times.}
Let $\theta>0$ and $Y>0$ be random variables whose joint distribution is defined by prior distribution $\theta\sim\GammaDist(2,2)$ and likelihood $Y|\theta\sim\Exp(\theta)$. $Y$ might represent the waiting time until a bus arrives, or the time until failure for a light bulb, or the time taken for a particle to decay. We can write down the corresponding densities: $$p_{\theta}(\theta)=4\theta e^{-2\theta} \quad\text{and}\quad p_{Y|\theta}(y|\theta)=\theta e^{-\theta y}.$$ 
We know that $\theta$ and $Y$ have a joint density on $\Re_+^2$; it is given by $p_{\theta,Y}(\theta,y)=p_\theta(\theta)p_{Y|\theta}(y|\theta)$. Therefore, the version of Bayes' theorem for continuous random variables can be used here, and we get that
\begin{align*}
    p_{\theta|Y}(\theta|y)&\propto p_\theta(\theta)p_{Y|\theta}(y|\theta)\\
    &\propto\theta^2e^{-(2+y)\theta}.
\end{align*}
This can be recognised as the density of the $\GammaDist(3,2+y)$ distribution, up to a normalising constant, and with the correct normalising constant the posterior density is
\begin{equation}
    p_{\theta|Y}(\theta|y)=\frac{(y+2)^3}{2}\theta^2 e^{-(2+y)\theta}.
    \label{eqn:examplebayesrule} 
\end{equation}

\subsubsection{Why is the Borel--Kolmogorov paradox not a problem here?}
\label{subsection:BK}
The Borel--Kolmogorov paradox only appears when one tries to condition on a zero-probability event (without reference to a random variable).

We avoid the Borel--Kolmogorov paradox by equipping the naked event $\{Y=y\}$ with the random variable $Y$. When we work with a conditional density $p_{X|Y}$, we do not attempt to condition directly on the zero-probability event $\{Y=y\}$. Rather, we condition on the random variable $Y$ to obtain the function $p_{X|Y}$; then, when a particular value $y$ is observed, we may evaluate that function at the value $y$. Bayes' rule also works because it concerns conditioning on random variables, not on zero-probability events.

The following two points emphasise the difference between conditioning on a random variable and conditioning on a zero-probability event.

\paragraph{The definition of conditional density is well-motivated.} It is interesting to note that we may motivate the definition (\ref{eqn:conddistdiscrete}) in the continuous case by requiring that conditioning on a continuous random variable agrees with conditioning on a positive-probability event:
\begin{itemize}
    \item On the one hand, we ought to be able to integrate the conditional density $p_{X|Y}(x|y)$ to obtain a conditional cumulative distribution function
    \begin{equation}
        F_{X|Y}(x|y)=\int_{-\infty}^{x}p_{X|Y}(x'|y)\der x'.
        \label{eqn:conditionalcdfA}
    \end{equation}
    \item On the other hand, (\ref{eqn:condprobevents}) tells us how to condition on the event that $Y$ lies in a small interval $[y-\delta y, y+\delta y]$:
    \begin{align}
        \P(X\leq x|Y\in[y-\delta y,y+\delta y])
        &=\frac{\P(X\leq x,Y\in[y-\delta y,y+\delta y])}{\P(Y\in[y-\delta y,y+\delta y])} \notag \\
        &=\frac{\int_{y-\delta y}^{y+\delta y}\int_{-\infty}^{x} p_{X,Y}(x',y') \der x' \der y'}{\int_{y-\delta y}^{y+\delta y} p_Y(y') \der y'}. \label{eqn:conditionalcdfB}
    \end{align}
\end{itemize}
For small $\delta y$, we might like (\ref{eqn:conditionalcdfB}) to be a good approximation for (\ref{eqn:conditionalcdfA}), and that they agree in the limit as $\delta y\to0$.
This does happen under suitable regularity conditions: assume that on the right-hand side of (\ref{eqn:conditionalcdfB}) the integrands $\int_{-\infty}^{x} p_{X,Y}(x',y') \der x'$ and $p_Y(y')$ are continuous in $y'$. Then, it follows that as the right-hand side of (\ref{eqn:conditionalcdfB}) converges to
${\int_{-\infty}^{x} p_{X,Y}(x',y) \der x'}/{p_Y(y)}$ as $\delta y\to0$, whenever $p_Y(y)>0$. Our definition of the conditional density (\ref{eqn:conddistdiscrete}) ensures this is equal to the right-hand side of (\ref{eqn:conditionalcdfA}).

Thus, we have been able to motivate the definition of the conditional density of $X$ given $Y$ by considering positive-probability events $\{Y\in[y-\delta y,y+\delta y]\}$ which converge to the zero-probability event $\{Y=y\}$. In some sense, when the random variable $Y$ is declared, events $\{Y\in[y-\delta y,y+\delta y]\}$ become a natural choice for positive-probability events that converge to $\{Y=y\}$.\footnote{An equally natural choice would be $\{Y\in[y-\delta y, y+2\delta y]\}$. It is important that these events are composed of level sets of $Y$; what is not important is how fast we are approaching $y$ from above and below. Changes in the parametrisation of $Y$ only correspond to changes in how fast these events approach $y$ from above and below, so the choice of events remains natural under reparametrisation (informally speaking).} By contrast, there is no natural choice for positive-probability events that converge to an arbitrary zero-probability event. This may give a little intuition as to why conditioning on a random variable is qualitatively different from conditioning on a zero-probability event.

\paragraph{The notion of $p_{X|A}$ is meaningless if $A$ is a zero-probability event.} To further emphasise the semantic difference between conditioning on a random variable and conditioning on a zero-probability event, consider the following. An attempt to define the conditional density of $X$ given zero-probability event $A$ only declares one density $x\mapsto p_{X|A}(x)$. In contrast, the conditional density of $X$ given random variable $Y$ declares a whole collection of densities $x\mapsto p_{X|Y}(x|y)$, one for each value of $y$. These latter conditional densities can be combined together to retrieve facts about positive probabilities; for instance, we can recover the fact that the joint distribution function of $(X,Y)$ is
$$\P(X\leq x, Y\leq y)=\int_{-\infty}^{y}\left(\int_{-\infty}^{x}p_{X|Y}(x'|y') \der x' \right)p_Y(y')\der y'.$$
However, there is no way to retrieve meaningful statements about positive probabilities from any notion of $p_{X|A}(x)$. This is a reason that we should not even hope to define it. 

Indeed, looking back to the conditional density given by Equation (\ref{eqn:examplebayesrule}) in our example of exponential waiting times, we see that the equation declares a whole collection of densities, one for each possible realisation of $Y$. This is conceptually distinct from what happens 
in Borel’s original example (conditioning on a great arc) and in MC’s tomographic example
(conditioning on $\{v_1 = v_2\}$; we shall discuss this example in Section \ref{subsection:tomographic}).

\subsubsection{Is the conditional density uniquely defined?}
\label{subsection:unique}
The conditional density $p_{X|Y}(x|y)$ is essentially uniquely defined, when it exists.

Note that different density functions $f(z)$ can define the same distribution over $Z$. For example, the standard normal distribution has density function
\begin{equation}
    f(z)=\frac{1}{2\pi}e^{-z^2/2}.
    \label{eqn:normalA}
\end{equation}
However, the density function
\begin{equation}
    \tilde{f}(z)=\begin{cases}
    \frac{1}{2\pi}e^{-z^2/2}, &\text{for }z\neq42,\\
    0, &\text{for }z=42,
\end{cases}
    \label{eqn:normalB}
\end{equation}
describes exactly the same distribution. When we want to calculate probabilities like $\P(Z>10)$, we get the same result regardless of which density function we chose to use. This is because when we integrate over $z$, these trivial differences are inconsequential. Expressions (\ref{eqn:normalA}) and (\ref{eqn:normalB}) encode exactly the same information about the distribution of random variable $Z$. This reveals that density functions are not uniquely defined, but they are ``essentially" uniquely defined.

In a somewhat similar sense, the conditional density $p_{X|Y}(x|y)$ is not uniquely defined, but it is ``essentially" uniquely defined. This is best illustrated by an example.

\paragraph{Example: exponential waiting times, revisited.} We may slightly tweak (\ref{eqn:examplebayesrule}) and define a slightly different collection of distributions by
\begin{equation}
    \tilde{p}_{\theta|Y}(\theta|y)=
    \begin{cases}\frac{(y+2)^3}{2}\theta^2 e^{-(2+y)\theta}, &\text{if } y\neq42,\\
    e^{-\theta}, &\text{if } y=42.\end{cases}
    \label{eqn:tweak}
\end{equation}
Here, among all possible $y>0$ we have picked an arbitrary value of $y$ and assigned to it an arbitrary distribution over $\theta$. Now, $\tilde{p}_{\theta|Y}$ is a valid but different version of the conditional density of $\theta$ given $Y$. In fact, we may tweak the distributions over $\theta$ defined by $(p_{\theta|Y}(\theta|y))_{y\in Y}$ at more values of $y$, as long as the tweaks we make are confined to a set $S\subset\Re$ such that $\P(Y\in S)=0$.

We appreciate that one may intuitively think that these different versions of the conditional density imply inconsistencies, but we emphasise that these tweaks are trivial, since the different versions still agree at almost every $y$. [To be precise, the distributions described by the densities agree for almost all $y$ with respect to $\P\circ Y^{-1}$.] Knowing the distribution of $Y$ already, the expressions (\ref{eqn:examplebayesrule}) and (\ref{eqn:tweak}) encode exactly the same information about the conditional distribution of $\theta$ given $Y$.\footnote{And if we want to be careful about the fact that there are different versions of the conditional density, we can interpret Equations (\ref{eqn:bayesthmdiscrete}, \ref{eqn:bayesrulediscrete}) as equations that hold for each version, for almost all $y$ with respect to $\P\circ Y^{-1}$.}

Moreover, the practitioner is able to get by without worrying about these subtleties in practice, since in reality by following Bayes' rule, they will end up with the more sensible-looking (\ref{eqn:examplebayesrule}). Even if another practitioner uses the somewhat daft expression (\ref{eqn:tweak}), there is no chance that $Y$ will realise a value at which the two practitioners will obtain conflicting posteriors.

Up to now, some readers may have been wondering, instead of getting the conditional density of $X$ given a zero-probability event $A$, what if we just get the conditional density of $X$ given the random variable $\I_{A}$ (which is the indicator function  taking value 1 if $A$ and 0 otherwise)? In this case, all versions of the conditional density $p_{X|\I_{A}}$ will agree on the distribution $p_{X|\I_{A}}(\cdot|0)$, but the distribution $p_{X|\I_{A}}(\cdot|1)$ will not be uniquely defined. So, this procedure will not allow us to give any meaning to the notion of $p_{X|A}$.

\subsection{Towards conditioning with more general random elements} \label{subsection:towards}
So far, we have only discussed the case where $X$ and $Y$ are both discrete with joint probability mass function $p_{X,Y}(x,y)$ and the case where $X$ and $Y$ are both real with joint probability density $p_{X,Y}(x,y)$ on $\Re\times\Re$. In fact, these can be regarded as special cases of a more general setting. [For correctness, we will use some measure-theoretic terminology in this subsection, but it will be kept within brackets, and informal ideas can be understood by ignoring what is in the brackets.]

In the continuous case, $p_{X,Y}(x,y)$ is a probability density function on $\Re\times\Re$ [with respect to the Lebesgue measure]. In the discrete case, the probability mass function $p_{X,Y}(x,y)$ can also be viewed as a probability density but on a discrete space of possible values like $\N\times\N$ [with respect to the counting measure]. Random elements living on other spaces can also have well-defined probability densities. [See Proposition \ref{prop:rn}.] We use the more general term random element here, since random variables typically refer to the special case where random elements are real-valued.

Suppose $X$ takes values in a [standard Borel] space $E_X$ and $Y$ takes values in a [measurable] space $E_Y$. Suppose further that $(X,Y)$ has a joint probability density $p_{X,Y}(x,y)$ on $E_X\times E_Y$ [with respect to a suitable product measure]. Then, we can still define the conditional density of $X$ given $Y$ as a ratio of densities using the formula (\ref{eqn:conddistdiscrete}) as before, whenever $0<p_Y(y)<\infty$. Here, the denominator in the ratio of densities, $p_Y(y)$, is the marginal density of $Y$, and it can be evaluated as the integral over $x$ of the joint density. It follows that we get a new version of Bayes' theorem, which still has the old formulae (\ref{eqn:bayesthmdiscrete}, \ref{eqn:bayesrulediscrete}) but is more general. [For Bayes' theorem, we should additionally have that $E_Y$ is a standard Borel space
. Note that assuming a measurable space is a standard Borel space rules out some pathological spaces but retains those measurable spaces that we wish to work with in practice.]

In short, the key takeaway here is that for random elements living in more general spaces, conditional densities exist whenever joint densities exist.
The rest of this subsection will be devoted to illustrating why the generalisation we have described helps us to perform Bayesian inference in practice.

Typically in Bayesian inference, our parameters $\theta$ and data $Y$ are more than just simple univariate quantities: $\theta$ lives in some [standard Borel] space $E_\theta$ and $Y$ lives in some [standard Borel] space $E_Y$. We would like to use Bayes' theorem, and indeed Bayes' theorem applies as long as $(\theta,Y)$ has a joint density on $E_\theta\times E_Y$ [with respect to a suitable product measure]. For example, if we have two continuous real parameters and $n$ continuous real observations, then Bayes' theorem applies as long as $(\theta,Y)$ has a joint density on $\Re^2\times\Re^n$ [with respect to the Lebesgue measure]. If the observations are instead non-negative integers, then Bayes' theorem applies as long as $(\theta,Y)$ has a density on $\Re^2\times\N^n$ [with respect to the product of Lebesgue and counting measures].

We do need to ensure that $(\theta,Y)$ has a joint probability density $p_{\theta,Y}(\theta,y)$ on $E_\theta\times E_Y$ [with respect to a suitable product measure], but this condition is typically easy to check. Suppose $p_{\theta,Y}(\theta,y)$ is defined by
\begin{equation}
    p_{\theta,Y}(\theta,y)=\pi(\theta)f(y|\theta),\label{eqn:jointdensity}
\end{equation}
where $\pi(\theta)$ is a continuous or piecewise continuous [or merely measurable] density on $E_\theta$ [with respect to a base measure $\mu$] and $f(y|\theta)$ is continuous or piecewise continuous [or measurable] in $(y,\theta)$ and a density on $E_Y$ [with respect to a base measure $\nu$] for each $\theta$; in such a case, the condition is satisfied [and $p_{\theta,Y}$ is a density with respect to product measure $\mu\otimes\nu$]. Then, the conditional densities exist, and Bayes' theorem can be used.

We can also allow for $\theta$ and/or $Y$ to be composed of both discrete and continuous entries; in this case, we need to build the likelihood and prior in an appropriate way (using probability mass functions for the discrete entries and probability density functions for the continuous entries). It is also possible to have that $\theta$ and $Y$ are trans-dimensional, for example by taking the (disjoint) union of two spaces $E_\theta=\Re^3\sqcup\Re^4$ for a parameter vector that is either three- or four- dimensional. We may also have cases where an individual parameter or data point has a distribution with both discrete and continuous components, which may arise when using a spike-and-slab prior that is a mixture of a point mass at zero and a continuous distribution centred around zero. [To handle this latter case, we would have to be careful to choose an appropriate dominating measure.]

In conclusion, we can safely work with conditional densities and Bayes' theorem in a wide range of cases. With the appropriate interpretation of the symbols involved, Bayes' theorem retains the same formula (\ref{eqn:bayesthmdiscrete}, \ref{eqn:bayesrulediscrete}).

What if joint densities do not exist? When joint densities do not exist, conditional densities may also not exist. Without conditional densities, our versions of Bayes' theorem are no longer available as a tool to help us compute the posterior numerically. Regardless, it turns out that conditional probabilities still do exist, and this leads to a definition for conditional distributions. Crucially, this ensures that objects like the posterior distribution still do exist even when joint densities do not exist (for example, when non-parametric processes are involved).

\subsection{Measure theory} \label{subsection:measuretheory}
Up to now, our presentation has been necessarily informal. To redress this, we shall now sketch a roadmap of the relevant theory, but the reader is invited to skip ahead to any of the later sections. We have put the main ideas in a sequence of boxes, with some signposting outside of these boxes. This theory was developed in the first half of the 20th century, and it refutes MC's claim that ``the notion of conditional densities is inadmissible".

Following \citet[ch.\ 6]{Kallenberg2010ed2}, we shall start with conditional expectation and conditional probability, defining what it means to condition on a $\sigma$-algebra and then what it means to condition on a random element. Then, we shall define conditional distributions. Finally, we shall frame conditional densities as one way of representing conditional distributions, and we shall show that this framing is consistent with the familiar formula (\ref{eqn:conddistdiscrete}).

Towards a general theory of conditioning, the first step is to define what it means to condition on a $\sigma$-algebra. The box below defines conditional expectations given a $\sigma$-algebra. There can be many versions of the conditional expectation of $X$ given $\mathcal{G}$, but these versions can only differ in trivial ways, in a way that is made precise by Proposition \ref{prop:condexp}.
\mybox{
    Let $(\Omega,\mathcal{F},\P)$ be a probability space throughout this subsection. Let $\mathcal{G}\subset\mathcal{F}$ be a sub-$\sigma$-algebra, and let $X$ be an integrable random variable. Then there exists a random variable $\E[X|\mathcal{G}]$, known as (a version of) the \textit{conditional expectation of $X$ given $\mathcal{G}$}, with two properties:
    \begin{enumerate}[label=(\roman*)]
        \item $\E[X|\mathcal{G}]$ is $\mathcal{G}$-measurable and integrable;
        \item For any $A\in\mathcal{G}$, we have
        $\E[\E[X|\mathcal{G}]\I_A] = \E[X\I_A]$.
    \end{enumerate}
    \begin{proposition}[see \citealt{Kallenberg2010ed2}, Theorem 6.1] \label{prop:condexp}
        Such a random variable always exists, and if another random variable $Z$ satisfies the definition, then $Z=\E[X|\mathcal{G}]$ almost surely.
    \end{proposition}
}
From this, we can define conditional probabilities given a $\sigma$-algebra, and this can be viewed as a special case of conditional expectation given a $\sigma$-algebra. 
\mybox{
    Let $A\in\mathcal{F}$ be an event. We define (a version of) the \textit{conditional probability of $A$ given $\mathcal{G}$} as the random variable
    $$\P(A|\mathcal{G})=\E[\I_A|\mathcal{G}].$$
}
Now we may define what it means to condition on a random element. The conditional expectation given a random element may be viewed as a special case of the conditional expectation given a $\sigma$-algebra. Likewise, the conditional probability given a random element may be viewed as a special case of the conditional probability given a $\sigma$-algebra.
\mybox{
    Let $Y$ be a random element on some measurable space $(E_Y,\mathcal{E}_Y)$, and let $\sigma(Y)$ denote the $\sigma$-algebra generated by $Y$. We define (a version of) the \textit{conditional expectation of $X$ given $Y$} by
    $$\E[X|Y]=\E[X|\sigma(Y)],$$
    and (a version of) the \textit{conditional probability of $A$ given $Y$} by
    $$\P(A|Y)=\P(A|\sigma(Y)).$$
}
Finally, the conditional distribution of a random element given another random element may be defined. Note that the probability distribution of $X$ is just a probability measure but over the appropriate state space (which we have called $E_X$) rather than over the sample space ($\Omega$); functionally, the probability distribution takes as input a measurable subset of the state space and returns as output a corresponding probability mass in $[0,1]$. Then, the conditional distribution of $X$ given $Y$ will be a collection of probability distributions over the state space $E_X$, one for each possible value $y$ of $Y$. There can be many versions of the conditional distribution of $X$ given $Y$, but any two versions can only differ on a trivial set of values $y$, as we see in Proposition \ref{prop:conddist}.

\mybox{
    Let $X$ be a random element on some standard Borel space 
    $(E_X,\mathcal{E}_X)$, and let $Y$ be a random element on some measurable space $(E_Y,\mathcal{E}_Y)$. Then there exists a map $\kappa(y,B)$ with signature $\kappa:E_Y\times \mathcal{E}_X\to[0,1]$, known as (a version of) the \textit{conditional distribution of $X$ given $Y$}, with three properties:
    \begin{enumerate}[label=(\roman*)]
        \item $\kappa(y,\cdot)$ is a probability measure on $E_X$ for each fixed $y\in E_Y$; 
        \item $\kappa(\cdot,B)$ is $\mathcal{E}_Y$-measurable for each fixed $B\in\mathcal{E}_X$;
        \item $\kappa(Y,\cdot)=\P(X\in\cdot\ |Y)$ almost surely.
    \end{enumerate}
    \begin{proposition}[see \citealt{Kallenberg2010ed2}, Theorem 6.3]
        \label{prop:conddist}
        Such a map always exists, and if another map $\tilde{\kappa}$ satisfies the definition, then $\kappa(Y,\cdot)=\tilde\kappa(Y,\cdot)$ almost surely.
    \end{proposition}
}

We have now stated the core objects in the theory of conditional probability. Further theory proves that these objects behave as we expect them to, that we can also perform iterated conditioning \citep[see][ch.\ 6]{Kallenberg2010ed2}, and that we can also consider $\sigma$-finite measures rather than just probability measures \citep[see][ch.\ 3]{Kallenberg2021ed3}. This last point is relevant to Bayesian inference as we sometimes work with improper priors \citep{chang1997conditioning}.

Now (departing from \citealt{Kallenberg2010ed2}), we turn to the special case where distributions have densities. In general, a probability density is the Radon--Nikodym derivative of a probability distribution, with respect to a some dominating measure (which may be referred to as a base measure or reference measure). Often the reference measure is implicitly the Lebesgue measure, but in general it need not be. Knowing the probability density (together with the reference measure) fully determines what the probability distribution is. The Radon--Nikodym theorem ensures that the probability density is well-defined:
\mybox{
    \begin{proposition}[Radon--Nikodym theorem, see \citealt{Kallenberg2010ed2}, Theorem 2.10] \label{prop:rn}
        Let $\mu,\nu$ be $\sigma$-finite measures on a measurable space $(S,\mathcal{S})$. If $\nu$ is absolutely continuous with respect to $\mu$ (denoted $\nu\ll\mu$), then there exists a $\mu$--almost everywhere unique measurable function $f:S\to[0,\infty)$ such that $\nu(A)=\int_A f\der\mu$ for all $A\in\mathcal{S}$.
    \end{proposition} 
    
    In general, $f$ is known as the Radon--Nikodym derivative of $\nu$ with respect to $\mu$; but if $\nu(S)=1$, then $\nu$ is a probability distribution, and $f$ is said to be the \textit{probability density of $\nu$ with respect to $\mu$}.
}

With this in mind, there is natural way to define the conditional density of $X$ given $Y$. Since the conditional distribution of $X$ given $Y$ can be viewed as a collection of distributions $\{\kappa(y,\cdot)\}_y$, one for each value $y$ of $Y$, we may view the conditional density of $X$ given $Y$ (with respect to a common reference measure) as a corresponding collection of densities $\{p_{X|Y}(\cdot|y)\}_y$. Such a conditional density exists exactly when there exists a reference measure that dominates all the distributions $\{\kappa(y,\cdot)\}_y$. In particular, this condition can be satisfied when the joint density of $X$ and $Y$ exists with respect to some reference measure that is a product measure on $E_X\times E_Y$, and in this case the formula (\ref{eqn:conddistdiscrete}) that we have been using all along holds. This formalises our key message from Section \ref{subsection:towards} that if joint densities exist, then conditional densities also exist. We make all the ideas in this paragraph precise in the box below; at last we see that the notion of conditional densities is supported by rigorous theory.

\mybox{
    A function $p_{X|Y}:E_X\times E_Y\to[0,\infty)$ is (a version of) the \textit{conditional density of $X$ given $Y$} with respect to a $\sigma$-finite reference measure $\mu$ if it is $(\mathcal{E}_X\otimes\mathcal{E}_Y)$-measurable and $\kappa(y,B) = \int_B p_{X|Y}(x|y)\der \mu(x)$ defines a version of the conditional distribution of $X$ given $Y$.
    
    \begin{proposition} \label{prop:densities}
        Suppose there exist $\sigma$-finite base measures $\mu$ on $E_X$ and $\nu$ on $E_Y$ such that the random element $(X,Y)$ has a density $p_{X,Y}(x,y)$ with respect to the product base measure $\mu\otimes\nu$. Let $p_{X}(x)$ denote (a version of) the marginal density of $X$ with respect to $\mu$. Then the following is a version of the conditional density of $X$ given $Y$:
    \begin{equation}
        p_{X|Y}(x|y)=
        \begin{cases}
            \frac{p_{X,Y}(x,y)}{\int_{E_X} p_{X,Y}(x',y)\der\mu(x')}, &\text{when the denominator is non-zero and finite},\\
            p_{X}(x), &\text{otherwise}.
        \end{cases}
        \label{eqn:formaldefn}
    \end{equation}
    \end{proposition}

    We shall give a sketch proof. Note that whenever we handle joint densities or conditional densities, it is helpful to bear in mind Tonelli's theorem, which ensures that the order of integration does not matter.
    \begin{proposition}[Tonelli's theorem, see \citealt{Billingsley1986}, Theorem 18.3]
        Let $(S,\mathcal{S},\mu)$ and $(T,\mathcal{T},\nu)$ be $\sigma$-finite measure spaces and $f:S\times T\to[0,\infty)$ be a ($\mathcal{S}\otimes\mathcal{T}$)-measurable function. Then $x\mapsto\int f(x,y)\der\nu(y)$ and $y\mapsto\int f(x,y)\der\mu(x)$ are $\mathcal{S}$- and $\mathcal{T}$-measurable, respectively, and
        $$\int f(x,y)\der(\mu\otimes\nu)(x,y)=\iint f(x,y)\der\mu(x)\der\nu(y)=\iint f(x,y)\der\nu(y)\der\mu(x).$$
    \end{proposition}

    \begin{proof}[Proof of Proposition \ref{prop:densities}]
        First, note that since the joint density $p_{X,Y}(x,y)$ is by definition non-negative and $(\mathcal{E}_X\otimes\mathcal{E}_Y)$-measurable, Tonelli's theorem applies to it, and we can check that $p_{X|Y}(x|y)$ as defined in (\ref{eqn:formaldefn}) is $(\mathcal{E}_X\otimes\mathcal{E}_Y)$-measurable. It remains to show that
        $$\kappa(y,B)=\int_B p_{X|Y}(x|y)\der\mu(x)$$
        is a version of the conditional distribution of $X$ given $Y$. Of the three defining properties that must be checked, (i) and (ii) are straightforward (using Tonelli). It remains to show (iii). Fix $B\in\mathcal{E}_X$. We need to show that
        $\kappa(Y,B)=\P(\{X\in B\}|Y)$ a.s.
        By Proposition \ref{prop:condexp}, the random variable $\P(\{X\in B\}|Y)$ is the a.s.-unique $\sigma(Y)$-measurable random variable such that
        \begin{equation}
            \E\left[\P(\{X\in B\}|Y) \I_{\{Y\in C\}}\right]=\E\left[\I_{\{X\in B\}} \I_{\{Y\in C\}}\right] \quad \text{for all }C\in\mathcal{E}_Y.
            \label{eqn:factforpropertyiii}
        \end{equation}
        But by using the fact that we have densities, using Tonelli again and bearing in mind that the denominator $\int_{E_X}p_{X,Y}(x',y)\der\mu(x')$ is non-zero and finite $(\P\circ Y^{-1})$-a.e., it can be seen that $\kappa(Y,B)$ is also a random variable satisfying (\ref{eqn:factforpropertyiii}), so it must be that $\kappa(Y,B)=\P(\{X\in B\}|Y)$ a.s., as required.
    \end{proof}

    The proposition above concerns existence. Uniqueness (in the almost-sure sense) is immediate from the uniqueness of the conditional distribution: if $p_{X|Y}$ and $\tilde{p}_{X|Y}$ are different versions of the conditional density, then $\tilde{p}_{X|Y}(\cdot|Y)$ and $p_{X|Y}(\cdot|Y)$ almost surely describe the same distribution over $E_X$. Note that when $p_Y(y)=0$ or $p_Y(y)=\infty$ the choice of $p_{X|Y}(\cdot|y)$ is arbitrary as the events $\{p_Y(Y)=0\}$ and $\{p_Y(Y)=\infty\}$ have zero probability of occurring; this justifies why we did not specify it in (\ref{eqn:conddistdiscrete}).
}

\subsubsection{Bibliographic notes}
Regarding conditional densities, relevant texts include \citet[ch.\ 14]{kingmantaylor1966probability}, \citet[ch.\ 6]{hoelportstone1971probability}, \citet[ch.\ 6]{ash1972probability}. These all discuss the conditional density of one continuous random variable with respect to another continuous random variable. The first two of these also mention the version of Bayes' theorem concerning conditional densities, and the last two of these note that the conditional densities may be extended to continuous random vectors. However, we did not find a text that makes explicit the fact that conditional densities may be extended to random elements living in more general spaces. We have chosen to present the more general case with random elements since, in practice, we may need to use Bayes' theorem in this case.

In probability textbooks, the Borel paradox is presented as a warning that we should not try to condition on zero-probability events. Relevant discussions are given in \citet[ch.\ 5]{Pollard2001}, \cite{chang1997conditioning}, and \cite{bungert2022lionatticresolution}. These provide further explanation as to why the concerns expressed by MC are insubstantial. MC has claimed that the Borel--Kolmogorov paradox is avoided by the construction described in \cite{MOSEGAARD2002inverse}; however, what they describe is an informal idea that is never made rigorous, and the idea is geometric and not probabilistic.\footnote{
    Moreover, the text of \cite{MOSEGAARD2002inverse} contains some issues.\begin{itemize}
        \item It makes the unjustified claim that their construction gives rise to ``a definition invariant under any change of variables". Their construction requires a notion of orthogonality, which is not invariant to all changes of variables.
        \item It states the postulate that given a ``volume measure" on a space, any other measure is absolutely continuous with respect to the volume measure. This is false, as a measure could have all of its mass concentrated at a single point, say.
        \item It introduces a notion of ``Jeffreys quantities" that gives a rule of thumb for choosing priors. Their approach is difficult to justify, as appropriate choice of prior ought to depend on the model used (and any domain knowledge); moreover, priors chosen in this way differ fundamentally from Jeffreys priors \citep{jeffreys1946prior}, which have the key property that in any parametrisation of the parameter space the same prior is obtained. 
    \end{itemize}
}
\cite{bungert2022lionatticresolution} also approaches the Borel--Kolmogorov paradox through a geometric lens; their text is based on the rigorous theory of Hausdorff measures \citep[ch.\ 19]{Billingsley1986}. However, these geometric ideas are not applicable to statistical inference; in probability, events should have meaning without needing to position them inside a Euclidean space.

Turning to the measure-theoretic foundations, we have mostly followed \citet[ch.\ 6]{Kallenberg2010ed2}. While this text is more terse, \citet[chs.\ 33--34]{Billingsley1986} is more readable. While the existing theory of conditional probability, conditional expectation and conditional distribution has been long-established, there are still different ways to present the subject:
\begin{itemize}
    \item \citeauthor{Kallenberg2010ed2} defines conditional probabilities as a special case of conditional expectations, whereas \citeauthor{Billingsley1986} defines conditional probabilities first and subsequently constructs conditional expectations.
    \item When we condition on a random variable $Y$, we may either define the conditional probability as a function of $y\in E_Y$ that is $(\P\circ Y^{-1})$-a.e.\ unique, or as a $\sigma(Y)$-measurable random variable that is $\P$-a.e. unique (that is, a.s.\ unique). \citet{Pollard2001} takes the former view, whereas \citeauthor{Kallenberg2010ed2} and \citeauthor{Billingsley1986} take the latter.
    \item \citet{Kallenberg2021ed3} presents conditional distributions as part of a more general theory of disintegration compared to \citet{Kallenberg2010ed2}.
\end{itemize}

\section{The modeller needs to specify an appropriate model} \label{section:3}
In this section, we deal with the examples that MC provides in their Appendices A, F and G. A loosely unifying theme of this section is that the issues that MC reports can be addressed by setting up appropriate models. In the subsections that follow, we shall address those examples, explaining that they fail to show genuine inconsistencies in Bayesian inference and describing the modelling approaches that MC could have taken to perform Bayesian inference in a suitable way.

\subsection{MC's Appendix G demonstrates that different models can yield different results}
Obviously, different models can yield different results, if those models are not equivalent.

In their Appendix G, MC demonstrates that models that are different (via fixing a different value of $\sigma$) can produce different results. While this is true, we think it should not be surprising and it should not be thought of as an inconsistency. In this situation, if the modeller does not know what the appropriate value of $\sigma$ is for their model, they may treat $\sigma$ as a parameter and assign to it a prior distribution reflecting the range of plausible values.

\subsection{MC's ``simple tomographic example" (in their Appendix A) is invalid: conditioning on a zero-probability event is not allowed}
\label{subsection:tomographic}
In the simple tomographic example (Appendix A of MC), there are two unknown parameters $\bm{v}=(v_1,v_2)$. There is a prior density over $\bm{v}$; they condition on the data to obtain the posterior density; then they condition on the event $\{v_1=v_2\}$. They then find that the resulting conditional distribution is ill-defined.

What has happened here? Using Bayes' theorem to condition on the data is fine. There is no inconsistency that arises when we condition on the data, even under different parametrisations of the parameter space. However, their conditioning on the zero-probability event $\{v_1=v_2\}$ is invalid. We reiterate that one must not try to condition on a zero-probability event; there needs to be a random variable to condition on. It can happen neither in a Bayesian methodology, nor in any other methodology.

In practice, it may be that we observe data that provides evidence about whether the two velocities are similar; then we can follow the usual Bayesian framework, with data being modelled probabilistically. As discussed in Section \ref{section:2}, we are allowed to condition on data when the data are treated as random.

Intuitively, MC has written down a model that knows that $v_1$ and $v_2$ are different, so it makes no sense to impose the restriction that $v_1$ and $v_2$ are equal; they are asking the model to do something it was never designed to do.

This begs the question, what should we do, as modellers, if we want to perform Bayesian inference for the velocities under the model assumption that $v_1=v_2$? In this case, we could have started with just one velocity parameter and placed a prior over that one parameter. Then Bayes' theorem determines the correct posterior given the data. Even if we had used the slowness parametrisation, we would have obtained the correct posterior.

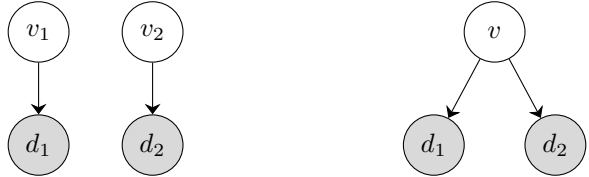
\begin{figure}[!ht]
    \centering
    \begin{tikzpicture}[
        every node/.style={circle, draw, minimum size=8mm},
        >={Stealth[length=4pt,width=6pt]}
    ]
    \node (v1) at (0,0) {$v_1$};
    \node[fill=gray!30] (d1a) at (0,-1.5) {$d_1$};
    \draw[->] (v1) -- (d1a);
    \node (v2) at (1.5,0) {$v_2$};
    \node[fill=gray!30] (d2a) at (1.5,-1.5) {$d_2$};
    \draw[->] (v2) -- (d2a);
    \node (v) at (6,0) {$v$};
    \node[fill=gray!30] (d1b) at (5.2,-1.5) {$d_1$};
    \node[fill=gray!30] (d2b) at (6.8,-1.5) {$d_2$};
    
    \draw[->] (v) -- (d1b);
    \draw[->] (v) -- (d2b);
    
    \end{tikzpicture}
    \caption{A model with two distinct parameters (left) is qualitatively different from a model with a shared parameter (right).}
    \label{fig:tomographic}
\end{figure}

A plausible scenario is that we are unsure whether or not $v_1=v_2$, but we have reason to believe that it may be the case. Then the obvious option is to fit two distinct models (as in Figure \ref{fig:tomographic}) to the data and use model comparison metrics to quantify to what extent the data support one model over the other. 
Another option that is available to us in the Bayesian framework is to fit a single model that incorporates both possibilities, a hierarchical model. This would entail introducing an additional parameter $q$ representing the probability that $v_1=v_2$, and placing a uniform prior on that parameter $q\sim\Unif[0,1]$, say.

Another plausible scenario is that we have reason to believe the velocities are similar, even if they are not exactly the same. This could be represented in a joint prior distribution over $(v_1,v_2)$ with positive correlation between the two.

In summary, Bayesian inference offers a flexible framework in which many models are possible, but the steps taken in MC's example are invalid.

\subsection{MC's examples concerning acausality (in their Appendix F) are invalid: their models have not been set up correctly}
The two examples given by MC in Appendix F, which are used to exhibit ``acausality", are invalid.

The conclusion they come to is that ``the computed prior distributions of [two parameters] vary with the forward relation"; but in actual fact they have presented (an attempt at computing) the posterior distribution, not the prior distribution. So, their conclusion should have been that the computed \textit{posterior} distributions of those parameters vary with the forward relation, and naturally this is to be expected.

It is difficult to make sense of their working. It seems that the data $d$, its mean given parameters $g(m)$, and the noise term $n=d-g(m)$ have all been conflated with each other. In addition, there may be a misunderstanding over what parameters are in Bayesian inference. MC says that there is only one parameter $m$, but they also talk of a prior over $d$ and introduce $\delta$ and $\lambda$, which are unobserved random variables. In actual fact, $d$ is an observed variable and should be regarded as data, while $m,\delta,\lambda$ are unobserved random variables and should be regarded as parameters, even if some of these are not parameters of interest.

Using our best guess at what was intended, we illustrate the correct approach. We start with the three parameters $\bm{\theta}=(m,\lambda,\delta)$, and we denote our single data point $y$. We draw a graph to illustrate the model in Figure \ref{fig:acausality}. Including such graphs is good practice, especially when we work with more complex hierarchical models than this, as it helps us to avoid getting confused when reasoning with priors, likelihoods and posteriors.

\begin{figure}[!ht]
    \centering
    \begin{tikzpicture}[
        every node/.style={circle, draw, minimum size=8mm},
        >={Stealth[length=4pt,width=6pt]}
    ]
    \node (delta) at (0,0) {$\delta$};
    \node (m) at (0,-1.5) {$m$};
    \node (lambda) at (1.5,-1.5) {$\lambda$};
    \node[fill=gray!30] (y) at (0.75,-3) {$y$};
    
    \draw[->] (delta) -- (m);
    \draw[->] (m) -- (y);
    \draw[->] (lambda) -- (y);
    
    \end{tikzpicture}
    \caption{A graphical illustration of the model. The shaded node represents an observed quantity, i.e.\ data. All unshaded nodes represent unknown quantities, i.e. parameters; one parameter ($\delta$) may be called a hyperparameter and two ($\delta,\lambda$) might be thought of as nuisance parameters. The dependencies in the model specification are shown with arrows.}
    \label{fig:acausality}
\end{figure}
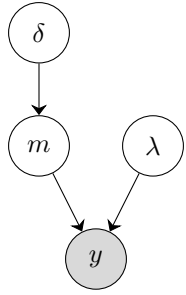

We have prior distributions
$$m|\delta\sim\Normal(m_0,\delta^2), \quad\delta\sim\Normal(0,\sigma_\delta^2), \quad\lambda\sim\Normal(0,\sigma_\lambda^2),
$$
implying the joint prior density
$$\pi(\bm{\theta}) 
= p_{m|\lambda}(m|\lambda)p_{\lambda}(\lambda)p_{\delta}(\delta) 
\propto \frac{1}{\delta}\exp\left(-\frac{1}{2}\left(\frac{(m-m_0)^2}{\delta^2}+\frac{\lambda^2}{\sigma_\lambda^2}+\frac{\delta^2}{\sigma_\delta^2}\right)\right).$$
Note that there is no reference to any data in the prior. We define the distribution of the data by
$$y\sim\Normal(km,\lambda^2)$$
which is equivalent to $\epsilon\sim\Normal(0,\lambda^2)$, $y=km+\epsilon$. This implies the likelihood is
$$f(y|\bm{\theta})=
\frac{1}{\sqrt{2\pi \lambda^2}}
\exp\left(
-\frac{(y - km)^2}{2\lambda^2}
\right).
$$

Having set up our model explicitly, we will now be able to use Bayes' theorem (\ref{eqn:bayesrule}) to find the posterior distribution, $\pi(m,\lambda,\delta|y)$, over the parameters $(m,\lambda,\delta)$.

\subsection{Discussion}
When performing Bayesian data analysis, it is helpful to bear in mind the sequence of steps that form the Bayesian workflow \citep{Gelman2020bayesianworkflow}, and this workflow can very roughly be condensed into three steps \citep{Gelman2013bda}:
\begin{enumerate}
    \item Set up a model that specifies the parameters $\theta$, data $X$, prior density $p(\theta)$ and likelihood $p(x|\theta)$;
    \item Upon observing a realisation of the data, compute the posterior density $p(\theta|x)$;
    \item Evaluate the model fit, and possibly decide to fit alternative models, returning to Step 1.
\end{enumerate}

The first step is important. It is the job of the modeller to set up the model in a way that is appropriate for their inference goals, whether that may be to estimate a parameter, or to measure the evidence for different hypotheses, or to build a predictive model, or something else. Choosing an appropriate model can be hard. Nonetheless, even if it is not possible to achieve a perfect representation of the real-world system, it may still be possible to model the system in a useful way.

\section{Bayesian inference remains invariant to reparametrisations of the data space} \label{section:4}
The remaining examples of MC are based on a change of variables in the data space. In each example, they perform a calculation using two different parametrisations of the data space and obtain inconsistent results. We find that the contradiction stems from their incorrect treatment of the change of variables. In fact, Bayesian inference does not depend on reparametrisations of the data space. We first demonstrate why this is the case; we then shed light on their error.

Consider again the basic framework where we have prior density $\pi(\theta)$, and suppose the data $Y$ is a random vector in $E_Y\subset\Re^n$ with density $f(y|\theta)$. (Recall that such densities can be defined in the way described in Section \ref{subsection:towards}.) The posterior is given by Equation (\ref{eqn:bayesrule}).

Now suppose there is a reparametrisation of our data, living in data space $E_X\subset\Re^n$, and $X=\phi(Y)$ where $\phi:E_Y\to E_X$ is a bijection with continuous partial derivatives
$J_{ij}(y)=\frac{\der\phi_i(y)}{\der y_j}$.
Then the density in this reparametrisation can be written as
$$\tilde{f}(x|\theta)=|J(\phi^{-1}(x))|^{-1}f(\phi^{-1}(x)|\theta),$$
where $J(y)=\bigl(J_{ij}(y)\bigr)_{1\le i,j\le n}$ is the Jacobian matrix and $|J(y)|$ denotes its determinant \citep[see][for a rigorous treatment]{Billingsley1986}. We say that the determinant of the Jacobian matrix is the Jacobian. Then, the posterior is given by
$$\pi(\theta|x)\propto\pi(\theta)\tilde{f}(x|\theta)=\frac{\pi(\theta)f(\phi^{-1}(x)|\theta)}{|J(\phi^{-1}(x))|}.$$
Next, the key observation is that the Jacobian does not depend on the parameters $\theta$; it only depends on the data. Therefore, it may be absorbed into the proportionality constant, and we obtain
$$\pi(\theta|x)\propto\pi(\theta)f(\phi^{-1}(x)|\theta).$$ Bearing in mind that $X=\phi(Y)$, this agrees exactly with what we already knew in Equation (\ref{eqn:bayesrule}). Thus, reparametrising the data space does not lead to any inconsistencies in the posterior.\footnote{It is worth pointing out that reparametrising the parameter space also does not lead to any inconsistencies in the posterior. To see this, suppose $\theta$ lives in $E_\theta\subset\Re^p$ and its alternative parametrisation is $\eta$ in $E_\eta\subset\Re^p$, related by the bijection $\psi: E_\eta\to E_\theta$. Then the priors will be related by $\check{\pi}(\eta)=\pi(\psi(\eta))|K(\eta)|$ and the posteriors will be related by $\check{\pi}(\eta|y)=\pi(\psi(\eta)|y)|K(\eta)|$, where $K$ denotes the Jacobian for $\psi$. Thus, if the priors are equivalent, then the posteriors will be equivalent. Beware that a uniform prior on $\theta$ is not in general equivalent to a uniform prior on $\eta$; the property of being ``uniform" is not invariant to reparametrisation.}

As a result, there are no inconsistencies when we use this posterior for Bayes factors; likewise in empirical Bayes there are no inconsistencies when we use a marginal likelihood to obtain point estimates of hyperparameters. 

Orthogonal to MC's argument, there is a genuine criticism of empirical Bayes, which is that it uses the data twice: once to estimate hyperparameters, and again to obtain the posterior under a model with fixed hyperparameters. This means that uncertainty in the hyperparameters has been ignored, and we are overfitting to the data. However, empirical Bayes can still be useful if fitting a fully hierarchical model with uncertainty in the hyperparameters would be computationally prohibitive. Also, when there is sufficient data, there should generally be little uncertainty in the hyperparameters anyway.

\subsection{MC's examples involving empirical Bayes and Bayes factors (their Appendices C and D) are invalid: reparametrisation is performed incorrectly}
What we say is incompatible with the claims of inconsistency that MC makes with their examples involving empirical Bayes and Bayes factors (detailed in their Appendices C and D, respectively). How, then, have they gone wrong? 

Essentially, they have used the wrong Jacobian. In the passage above, we have made explicit the fact that the Jacobian $J(y)$ depends on the data $y=\phi^{-1}(x)$, but not on observed parameters. However, they have evaluated the Jacobian not at the observed data $y$ but at the value predicted by the forward model $g(m)$. This is invalid, but the error has been obfuscated by their ``Formulation 2", which is not a fully precise formulation.

We remark that there is no benefit to a practitioner in using ``Formulation 2", and therefore it is unlikely for this error to be made in practice. Nonetheless, we shall attempt to make the error more obvious to the more interested reader by first introducing more precise notation to express their ``Formulation 2".


Generally, we are considering a particular class of model in which the data has the form $y=g(m)+\epsilon$, where $m$ consists of unknown parameters and $\epsilon$ has distribution possibly depending on unknown parameters $\theta_\epsilon$. Together, the parameters are $\theta=(m,\theta_\epsilon)$. At this point, the likelihood has the functional form $$f_{Y|\theta}(y|\theta)=h(y;g(m),\theta_\epsilon).$$
We additionally suppose that $\epsilon$ has a symmetrical distribution (or in symbols, $\epsilon\stackrel{d.}{=}-\epsilon$).
This has the consequence that the likelihood has the functional form
\begin{equation}
    f_{Y|\theta}(y|\theta)=h(y;g(m),\theta_\epsilon)=h(g(m);y,\theta_\epsilon),
    \label{eqn:likelihoodfunctionalform}
\end{equation}
and therefore the posterior has the functional form
\begin{equation}
    \pi(\theta|y)\propto\pi(\theta)h(g(m);y,\theta_\epsilon).
    \label{eqn:bayesrulefunctionalform}
\end{equation}

Now we are ready to perform the change of variables in the data space, as is attempted in Appendices C and D of their manuscript. Suppose (as before) that there is a reparametrisation of our data, living in data space $E_X\subset\Re^n$, and $X=\phi(Y)$ where $\phi:E_Y\to E_X$ is a bijection with continuous partial derivatives $J_{ij}(y)=\frac{\der\phi_i(y)}{\der y_j}$. Since the likelihood in the $y$-parametrisation has the form (\ref{eqn:likelihoodfunctionalform}), the likelihood in the $x$-parametrisation has the form 
\begin{equation}
    \tilde{f}_{X|\theta}(x|\theta)=\frac{h(\phi^{-1}(x);g(m),\theta_\epsilon)}{|J(\phi^{-1}(x))|}=\frac{h(g(m);\phi^{-1}(x),\theta_\epsilon)}{|J(\phi^{-1}(x))|},
    \label{eqn:transformedlikelihoodfunctionalform}
\end{equation}
leading to the appropriate posterior (\ref{eqn:bayesrulefunctionalform}).

However, one might na\"ively think that Equation (\ref{eqn:likelihoodfunctionalform}) expresses a rule that ``whenever we see $x$ and $g(m)$ we may swap them around", in which case we will end up with the Jacobian being evaluated at $g(m)$ instead of $\phi^{-1}(y)$, as has happened in the examples of MC's manuscript. A potential source of confusion is in Equation (5) of their manuscript, where the dependence on the data has silently been omitted from the right-hand side; throughout the rest of the text, the dependence of the Jacobian on the data has been omitted.

\section{What are some genuine advantages and disadvantages of using Bayesian inference?} \label{section:5}
One barrier to adopting a Bayesian approach is that it requires a certain level of statistical literacy. We point out that there are a number of very accessible textbooks for Bayesian inference that may help to overcome this barrier, including \cite{Gelman2013bda}, \cite{lambert2018student}, and \cite{McElreath2020rethinking}.

Besides this, when deciding whether to use Bayesian inference in a statistical problem, there are a number of factors to consider. In this stand-alone section, we would like to direct attention towards legitimate issues.

\paragraph{Philosophical interpretation.} As individuals, we each possess uncertainties about the world around us. These uncertainties are important, as they inform how we make decisions. As we go about the world making observations and learning information, we update our beliefs. The Bayesian framework formalises this intuition. This is a compelling philosophical reason to choose a Bayesian approach.
On the other hand, in favour of frequentist statistics, one might argue that parameters should be viewed as fixed but unknown features of the world, rather than treating parameters as random entities. In addition, one might object to the fact that Bayesian inference is subjective, since results depend on the subjective choice of prior. 

\paragraph{Uncertainty quantification.} A key advantage of the Bayesian paradigm is how it provides uncertainty quantification. Bayesian inference is able to answer questions like ``in light of this data, to what degree do we believe that hypothesis $A$ holds?" and make statements like ``under our model, there is a 95\% chance that the true parameters lie in region $R$". These are arguably more natural and interpretable than their frequentist counterparts (``if hypothesis $A$ were true, what would be the probability of observing data more extreme than this?" and ``the procedure used to construct region $R$ has 95\% coverage probability"). 
The Bayesian framework also makes it straightforward to incorporate unobserved latent variables, missing data, and learnable hyperparameters; these are all modelled using additional parameters, and uncertainty in these parameters is naturally propagated through the model. In addition, even when data are very scarce and there is very high uncertainty as a result, Bayesian inference can provide principled uncertainty quantification. 
In spite of all this, we should bear in mind that Bayesian inference cannot find the ground truth parameters if the model was misspecified to begin with; the resulting posterior distribution may confidently point to the wrong parameter values. The results may still be useful, as the posterior mass will concentrate towards the parameters that give the best possible fit to the data. (This is why in Section \ref{section:3} we highlighted that the modeller should take care to choose an appropriate model. Of course, this issue is not unique to Bayesian inference.)

\paragraph{Prior specification.} In some cases, the ability to specify a prior can be attractive, as it allows us to encode prior information coming from domain knowledge or previous experiments. Informative priors can also provide a natural source of regularisation, favouring parsimony and guarding against overfitting in more complex models. However, in practice, specifying an appropriate prior can be very challenging, especially when there is very little prior information, and this blessing turns into a curse. Different choices of priors will, through Bayes' rule, lead to different posteriors; this issue is known as prior sensitivity. Prior sensitivity can be particularly severe when performing model comparison with Bayes factors, whereas predictive approaches based on the posterior predictive distribution (including PSIS-LOO for model comparison) tend to be more robust. It is interesting to note that for a broad class of parametric models, the Bayesian approach agrees with the frequentist approach as the amount of data grows to infinity, and the effect of the prior washes out.

\paragraph{Computational difficulties.} Up to now, we have not discussed computational considerations. For simple models like our exponential waiting times example, the posterior distribution can be derived analytically. However, for more complex models, analytic solutions are typically unavailable. In such cases, direct evaluation of the posterior density requires computation of the normalising constant in Bayes’ theorem. This normalising constant is the marginal likelihood $p_Y(y)=\int \pi(\theta)f(y|\theta) \der \theta$ which involves integrating over a potentially high-dimensional parameter space; this is often computationally intractable. One solution is to sample parameters from the posterior distribution using a Markov chain Monte Carlo (MCMC) algorithm. Over the past few decades, developments in MCMC and software implementations have allowed practicable Bayesian inference. Even so, MCMC methods can still be computationally intensive. Alternative approaches can reduce the computational burden, at the expense of only targeting an approximation of the posterior. MCMC may also struggle to explore the parameter space efficiently if there are pathologies in the shape of the posterior function.

\section{Conclusions} \label{section:6}
We state the key conclusions of this note.

\begin{enumerate}
    \item Readers should be assured that the claims of \cite{Mosegaard2024} are invalid. Bayesian inference remains a valid and well-principled approach to a range of physical inverse problems.
    \item One important flaw in MC's argument was in trying to condition on a zero-probability event. We stress that this cannot be done, but what can be done is conditioning on a random variable.
    \item Probability densities can be defined not just for continuous random variables, but for more general random elements too. In practice, we do work with these more general densities: our prior functions, posterior functions, and likelihood functions are all probability densities over an appropriate space.
    \item When joint densities exist, so too do conditional densities. The notion of conditional density is firmly supported by theory.
    \item While Bayesian inference is a flexible approach providing principled uncertainty quantification, it does come with challenges, including prior sensitivity and computational difficulties.
\end{enumerate}

\section*{Contributors}
AY and BL conceived of this article. AY wrote the original draft. All authors contributed to reviewing and editing. We would also like to thank Kitty Knight and Alice Luo for their thoughtful feedback on earlier drafts.

\bibliography{ref.bib}

\appendix

\section*{October 2026: reply to \cite{Mosegaard2026response}}
When we originally uploaded this note on 27 March 2026, we had two main intentions: to explain that conditional densities are not ill-defined when handled appropriately, and to record issues we identified with \cite{Mosegaard2024}. On 30 April 2026, the authors uploaded their response \citep{Mosegaard2026response}. We thank them for responding. Having considered their response, we find that problems with their main arguments remain. Here we reply to specific points and correct one error of our own.

\begin{enumerate}
    \item Their response highlights that the process of conditioning depends on the choice of sub-$\sigma$-algebra. However, when we condition on an observation, that observation is associated with an observation process --- a data-generating process that models the observation as a random variable --- and this fixes the $\sigma$-algebra, so no ambiguity remains. 
    \item For example, we might observe which ``meridian circle" the meteorites will land on, and this fixes the sub-$\sigma$-algebra. Alternatively, the observation process might be that we observe whether or not the meteorites will land on one particular meridian circle. This might be a reasonable observation process if we believed that there was something special about that meridian circle making landing on that meridian circle a positive probability event to begin with. In either case, conditional probability tells us how to update the posterior without inconsistencies.
    \item It is true that the method they demonstrate, restricting the probability density function onto a submanifold and renormalising, is not a valid way to obtain Bayesian posteriors. But it is not a flaw in the Bayesian framework, as that method is not part of the Bayesian framework and the theory of conditioning does not tell us to do it.
    \item They assert ``the demonstration in G.0.1 should be combined with the subsequent demonstration G.0.2". The two sub-results do not combine into the advertised conclusion. G.2 merely presents the Rosenblatt transformation, which gives us a way to change variables while preserving probabilities. The evidence changes in G.1 because the model is changed there, not because coordinates are relabelled in G.2.
    \item In Appendices C and D, they showed that ``Formulation 2 gives the same results as Formulation 1"; this establishes that both rely on similar reasoning, not that the reasoning is valid. In Equation (\ref{eqn:likelihoodfunctionalform}) we provide the correct likelihood in a transformation of the data space that preserves probabilities; the forward model is already accounted for.
    \item In Appendix F, the prior on parameters and hyperparameters does not depend on the data and forward model, only the posterior on parameters and hyperparameters depends on the data and forward model. This is to be expected; no ``acausality" is demonstrated in their Appendix F.
\end{enumerate}

\paragraph{Corrigendum.} We regret that in the initial version of this note, we grouped MC's Appendix B together with their Appendices C and D. In fact, what we wrote in Section \ref{section:4} is only relevant to their Appendices C and D. The response of \cite{Mosegaard2026response} brings this error to our attention. The updated version of this note removes all reference to Appendix B. Regarding its content, we do not dispute the fact that the MAP estimate depends on the chosen model parametrisation. However, MAP non-invariance is only a defect of using the mode of a density as a point estimate, not of using the posterior distribution for inference; it does not imply the Bayesian framework itself is flawed.

\end{document}